\documentclass{amsart}
\usepackage{amssymb}
\usepackage{amsthm}
\usepackage{amsmath}
\usepackage{algorithm}
\usepackage{algpseudocode}
\usepackage{graphicx}

\newtheorem{theorem}{Theorem}[section]
\newtheorem{lemma}[theorem]{Lemma}
\newtheorem{proposition}[theorem]{Proposition}

\title{On Gr\"obner Basis Detection for Zero-dimensional Ideals}

\author[P. Ananth]{Prabhanjan Ananth}
\thanks{} 
\address{Dept. of Computer Science and Automation, Indian Institute of
Science}
\email{prabhanjan@csa.iisc.ernet.in}
\author[A. Dukkipati]{Ambedkar Dukkipati}
\thanks{} 
\address{Dept. of Computer Science and Automation, Indian Institute of
Science}
\email{ambedkar@csa.iisc.ernet.in}
\date{}

\begin{document}
\maketitle

\begin{abstract}
\noindent The Gr\"obner basis detection (GBD) is defined as follows: Given a set of polynomials, decide whether there exists -and if ``yes" find- a term order such that the set of polynomials is a Gr\"obner basis. This problem was shown to be NP-hard by Sturmfels and Wiegelmann. We show that GBD when studied in the context of zero dimensional ideals is also NP-hard. An algorithm to solve GBD for zero dimensional ideals is also proposed which runs in polynomial time if the number of indeterminates is a constant. 
\end{abstract}


\section{Preliminaries}
 
Consider the polynomial ring $P=K[x_1, \ldots ,x_n]$ where $K$ be a field. Let the set of all monomials be $\mathbb{T}^n=\{x_1^{\alpha_1} \cdots x_n^{\alpha_n}\ |\ (\alpha_1, \ldots ,\alpha_n) \in \mathbb{N}^{n}\}$. The leading term of a polynomial $f$ with respect to a term order is denoted by $lt(f)$. Let $\mathcal{F}$ be a set of polynomials. The set of leading terms of polynomials in $\mathcal{F}$ with respect to a term order $\prec$ is denoted by $LT_{\prec}(\mathcal{F})$. A pure power is a term which is of the form $x_i^\alpha$ for some $i \in \{1, \ldots ,n\}$ where $\alpha \in {\mathbb{N}}^{n}$. Any term order $\prec$ can be represented by a positive weight vector $w \in \mathbb{R}^{n}_{+}$ \textit{i.e.,}\\
$$X^{\alpha} \prec X^{\beta} \Leftrightarrow w^{T}\alpha < w^{T}\beta $$ 
A polynomial $f$ is said to be reduced to $h$ in one step by $g$ with respect to term order $\prec$, denoted by $f \xrightarrow{g} h$, if $lt(f)=lt(g)t$ and $h=f-tg$ for some term $t$. A polynomial $f$ is said to be reduced to $h$ by a set of polynomials $\mathcal{F}=\{f_1, \ldots ,f_s\}$ if $f \xrightarrow{f_{i_1}} h_1 \cdots \xrightarrow{f_{i_r}} h_r=h$ where $f_{i_1}, \ldots ,f_{i_r} \in \mathcal{F}$.\\

\noindent Let $f_1, \ldots ,f_s \in P$ and let $I=\langle f_1, \ldots , f_s \rangle$. Let $S$ be the system of polynomial equations

$$f_1(x_1, \ldots ,x_n) = 0$$
$$\vdots$$
$$f_s(x_1, \ldots ,x_n) = 0$$

\noindent The following result from commutative algebra will be useful later in our reductions.
\begin{proposition}
\label{zerodim}
Let $\sigma$ be a term ordering on $\mathbb{T}^n$. The following conditions are equivalent.\\
a) The system of equations $S$ has only finitely many solutions.\\
b) For every $i \in \{1, \ldots ,n\}$, there exists a number $\alpha_i \geq 0$ such that we have $x_i^{\alpha_i} \in LT_{\sigma}(I)$.\\
\end{proposition}

\noindent The proof of the above proposition can be found in~\cite{CCA1}. In the next section we prove the GBD result for zero-dimensional ideals and describe an algorithm to solve the problem.


\section{Gr\"obner basis detection for zero-dimensional ideals}

\noindent The Gr\"obner basis detection (GBD) problem was introduced by Gritzmann and Sturmfels in~\cite{MinkowskiAdd:GritzSturmfels} as an application of the Minkowski addition of polytopes. GBD is defined as follows:
\begin{quote}
\textbf{(GBD)} Given a set of polynomials $\mathcal{F}=\{f_1, \ldots ,f_s\}$, decide whether there exists -and if ``Yes" find- a term order $w \in \mathbb{R}^{n}_{+}$ such that $\mathcal{F}$ is a Gr\"obner basis with respect to $w$. 
\end{quote}

\noindent This problem was shown to be NP-hard by showing the NP-completeness of a variant of GBD called `Structural Gr\"obner basis detection' (SGBD). The SGBD is described as follows:

\begin{quote}
\textbf{(SGBD)} Given a set of polynomials $\mathcal{F}=\{f_1, \ldots ,f_s\}$, decide whether there exists -and if ``Yes" find- a term order $w \in \mathbb{R}^{n}_{+}$ such that $LT_{w}(\mathcal{F})$ is a set of pairwise coprime terms. 
\end{quote}

\noindent The main aim of this paper is to show that it is NP-hard to detect whether a set of polynomials is a Gr\"obner basis of a zero-dimensional ideal. The Gr\"obner basis detection zero-dimensional ideals is defined as follows.

\begin{quote}
\textbf{(GBD$_{0dim}$)} Given a set of polynomials $\mathcal{F}=\{f_1, \ldots ,f_s\}$, decide whether there exists- and if "Yes" find- a term order $w \in \mathbb{R}^{n}_{+}$ such that $\mathcal{F}$ is a Gr\"obner basis of a zero-dimensional ideal with respect to $w$.
\end{quote}

In order to show that GBD$_{0dim}$ is NP-hard, we define two problems HGBD$_m$ and HSGBD$_m$ which are variants of GBD and SGBD and determine their complexities.\\

\begin{quote}
\textbf{(HSGBD$_m$)} Given a set of homogenous polynomials $\mathcal{F}=\{f_1, \ldots ,f_s\}$ of constant degree $m$, decide whether there exists- and if "Yes" find- a term order $w \in \mathbb{R}^{n}_{+}$ such that $LT_{w}(\mathcal{F})$ is a set of pairwise coprime monomials.
\end{quote}

\begin{quote}
\textbf{(HGBD$_m$)} Given a set of homogenous polynomials $\mathcal{F}=\{f_1, \ldots ,f_s\}$ of constant degree $m$, decide whether there exists- and if "Yes" find- a term order $w \in \mathbb{R}^{n}_{+}$ such that $\mathcal{F}$ is a Gr\"obner basis with respect to $w$. 
\end{quote}

We first show that HSGBD$_m$ is NP-complete by a reduction from $m$-set packing. This reduction is obtained by a modification of the reduction from set packing to SGBD in~\cite{GBD}. The $m$-Set packing is described as follows.

\begin{quote}
\textbf{($m$-Set packing)} Given a family $S=\{S_1, \ldots ,S_k\}$ of subsets of $\{1, \ldots ,\nu\}$ such that all subsets have atmost $m$ elements, and a goal $c \in \mathbb{N}$. Are there $c$ pairwise disjoint sets in $S$?
\end{quote}

\noindent This problem is proved to be NP-complete (See, for example, in~\cite{CompIntract:Garey}). Without loss of generality, we can assume that there exists atleast two sets which are mutually disjoint.
\par We then show that HGBD$_m$ is NP-hard. We perform a polynomial time reduction from HGBD$_m$ to GBD$_{0dim}$ which will prove our result.


\subsection{Complexity}
\ \\ \textbf{Reduction from $m$-Set packing to Homogenous $m$-SGBD}:  The reduction from Set packing to SGBD is described in~\cite{GBD}. We tweak their approach to show that even homogenous SGBD  \textit{i.e.,} $\mathrm{HSGBD}_m$ is NP-complete. The modified reduction is described below.
\par Let $(\nu,S,c)$ be an instance of $(m-1)$-Set packing problem, we construct an instance of $\mathrm{HSGBD}_m$ as follows. Consider the polynomial ring\\
$$K[X_1, \ldots ,X_\nu,Y_{11}, \ldots ,Y_{1k}, \ldots ,Y_{c1}, \ldots , Y_{ck}]$$
in $\nu+ck$ variables, and we encode $S_j$ by the monomial $M_j=\prod_{i \in S_j}X_i$. Then we define $c$ polynomials
$$f_1=\sum_{j=1}^{k} Y_{1j}^{\alpha_{1j}}M_j, \ldots ,f_c=\sum_{j=1}^{k} Y_{cj}^{\alpha_{cj}}M_j,$$
where $\alpha_{ij}=(m-1)+1-\mathrm{deg}(M_j)$. Note that all the terms in the polynomials $f_1, \ldots ,f_s$ are of degree exactly $m$. Also, $\mathrm{deg}(M_j)$ is atmost $m-1$ and hence the exponent of $Y_{ij}$ is nonzero.\\ 
 
\begin{lemma}
$\mathcal{F}=\{f_1, \ldots ,f_m\}$ is a structural Gr\"obner basis if and only if $(\nu,S,c)$ is a ``Yes"-instance of the Set packing problem.
\end{lemma}
\begin{proof}
Let $\mathcal{F}=\{f_1, \ldots ,f_m\}$ be a structural Gr\"obner basis with leading terms $Y_{1{i_1}}^{\alpha_{1{i_1}}}M_i, \ldots ,Y_{m{i_m}}^{\alpha_{m{i_m}}}M_{m{i_m}}$. Then $M_{i_1}, \ldots ,M_{i_m}$ must have disjoint support, and the $m$ sets $S_{i_1}, \ldots ,S_{i_m}$ are disjoint. 
\par Let $S_{i_1}, \ldots ,S_{i_m}$ be disjoint subsets of $\{1, \ldots ,\nu\}$ in $S$. A weight vector $w \in \mathbb{N}^{\nu+mk}$ is defined to be 1 for all indeterminates except for $Y_{1{i_1}}, \ldots ,Y_{m{i_m}}$, which get weight $m+1$. Then the leading terms of $f_1, \ldots ,f_m$ with respect to $w$ are $Y_{1{i_1}}^{\alpha_{1{i_1}}}M_i, \ldots ,Y_{m{i_m}}^{\alpha_{m{i_m}}}M_{m{i_m}}$. Since they are pairwise coprime, $\mathcal{F}$ is a structural Gr\"obner basis with respect to $w$. The proof is complete.
\end{proof}

SGBD was used in~\cite{GBD} to show that GBD was NP-hard. The same proof also shows that $\mathrm{HGBD}_m$ is NP-hard. For completeness sake, we reproduce the proof here.

\begin{lemma}
HGBD is NP-hard.
\end{lemma}
\begin{proof}
Consider $\mathcal{F}$ to be the set of polynomials which is the output of the reduction from $m$-set packing to $m$-HSGBD.
\par Assume that there exists a term order $\prec$ such that all the leading terms of the polynomials of $\mathcal{F}$ are mutually coprime. This implies that $\mathcal{F}$ is a Gr\"obner basis with respect to $\prec$.
\par Assume that $\mathcal{F}$ is a Gr\"obner basis with respect to the term order $\prec$. Then, it needs to be shown that $lt(f_i)$ and $lt(f_j)$ are coprime for all $i$ and $j$. The $S$-polynomial of any two polynomials $f$ and $g$ reduces to zero with respect to $\mathcal{F}$. Any polynomial $f_k$ for $k \neq \{i,j\}$ involves a variable $Y_{lk}$ in its leading term and hence it does not participate in the reduction of $S(f_i,f_j)$. Thus $S(f_i,f_j)$ reduces to zero by $\{f_i,f_j\}$ only. Hence from Lemma 3.3.1 in \cite{adams1994introduction}, $lt(\frac{f}{d})$ and $lt(\frac{g}{d})$ are relatively prime where $d=\mathrm{gcd}(f,g)$. But since $\mathrm{gcd}(f,g)=1$, $lt(f)$ and $lt(g)$ are mutually coprime.
\end{proof}

\noindent \textbf{Reduction from $\mathrm{HGBD}_m$ to $\mathrm{GBD}_{0dim}$}:\\
Let $\mathcal{F}$ be the input to the homogenous $m$-SGBD. We will construct $\mathcal{F}'$ as follows. Let $\mathcal{G}= \{t\ |\ \mathrm{deg}(t)=2m+1\}$. Then,
$$\mathcal{F}'=\mathcal{F} \cup \mathcal{G}$$

\begin{theorem}
$\mathcal{F}$ is a Gr\"obner basis with respect to term order $\prec$ iff $\mathcal{F}'$ is a Gr\"obner basis of a zero-dimensional ideal with respect to $\prec$. 
\end{theorem}
\begin{proof}
Suppose $\mathcal{F}$ is a Gr\"obner basis then we show that $\mathcal{F}'$ is a Gr\"obner basis. For that we show that for any two polynomials $f,g \in \mathcal{F}'$, $S(f,g) \xrightarrow{\mathcal{F}'} 0$. \\
Case (i) $f,g \in \mathcal{F}$: Since $\mathcal{F}$ is a Gr\"obner basis \textit{w.r.t} $\prec$, $S(f,g) \xrightarrow{\mathcal{F}} 0$ which implies that $S(f,g) \xrightarrow{\mathcal{F}'} 0$.\\
Case (ii) $f \in \mathcal{G}, g \in \mathcal{G}$: Since $f,g$ are just monomials, we have $S(f,g)=0$.\\
Case (iii) $f \in \mathcal{F}, g \in \mathcal{G}$: Observe that the degree of $\mathrm{lcm}$ of two terms is greater than or equal to the maximum of the degrees of the two terms. Hence, $\mathrm{lcm}(lt(f),lt(g)) \geq 2m+1$. And so, the degree of $\frac{\mathrm{lcm}(lt(f),lt(g))}{lt(f)}$ is greater than or equal to $m+1$. Consequently, the degree of all the terms in ${\frac{\mathrm{lcm}(lt(f),lt(g))}{lt(f)}}f$ is greater than or equal to $2m+1$. Now, consider $S(f,g)$:
$$S(f,g)={\frac{\mathrm{lcm}(lt(f),lt(g))}{lt(f)}}f+{\frac{\mathrm{lcm}(lt(f),lt(g))}{lt(f)}}g$$
As argued earlier all the terms in the first part of the above sum have degree atleast $2m+1$ and the second part is a term of degree atleast $2m+1$. And hence, all the terms in $S(f,g)$ have degree atleast $2m+1$. It can be observed that $S(f,g)$ can be reduced by the polynomials in $\mathcal{G}$ \textit{i.e.,} $S(f,g) \xrightarrow{\mathcal{G}} 0$. Hence, $\mathcal{F}$ is a Gr\"obner basis \textit{w.r.t} $\prec$. Since $x_1^{2m+1}, \ldots ,x_n^{2m+1} \in \mathcal{G}$ and hence in $LT(\mathcal{F}')$, $\mathcal{F}$ is a Gr\"obner basis of a zero-dimensional ideal.\\     

\noindent Suppose that $\mathcal{F}'$ is a Gr\"obner basis. If we show that for any pair of polynomials $f,g$ in $\mathcal{F}$, $S(f,g) \xrightarrow{\mathcal{F}} 0$ then it would be imply that $\mathcal{F}$ is a Gr\"obner basis. Let $f,g \in \mathcal{F}$. The lcm of two terms divides the product of those two terms. Hence, the degree of lcm of two terms is atmost the  sum of degrees of the two terms. This implies that $\mathrm{lcm}(lt(f),lt(g))$ is atmost $2m$. Hence, degree of $\frac{\mathrm{lcm}(lt(f),lt(g))}{lt(g)}$ and $\frac{\mathrm{lcm}(lt(f),lt(g))}{lt(f)}$ is atmost $m$. Consequently, total degree of all terms in ${\frac{\mathrm{lcm}(lt(f),lt(g))}{lt(g)}}f$ and ${\frac{\mathrm{lcm}(lt(f),lt(g))}{lt(f)}}g$ is atmost $2m$. Hence, total degree of all the terms in $S(f,g)$ is atmost $2m$. The following claim proves that $S(f,g)$ can be reduced to zero only by the polynomials in $\mathcal{F}$. \\

\noindent \textbf{Claim}. $S(f,g) \xrightarrow{g_1} h_1 \xrightarrow h_2 \xrightarrow{g_3} \cdots \xrightarrow{g_r} h_r=0$ where $g_{i} \in \mathcal{F}$. Then, $h_i$ contains terms of degree atmost $2m$.\\
\textit{Proof}. We prove it by induction on the number of reduction steps. The assertion is true when $S(f,g)$ is reduced to zero in one step. Assume that after $l$ number of reduction steps, all the terms in $h_l$ have degree atmost $2m$. Now consider the ${(l+1)}^{th}$ reduction step which is $h_l \xrightarrow{g_{l+1}} h_{l+1}$. $h_l$ can be reduced only by polynomials in $\mathcal{F}$ and hence $g_{l+1} \in \mathcal{F}$. If $lt(h_l)=t.lt(g_{l+1})$ then $\mathrm{deg}(t)=m$ and hence, $h_{l+1}=h_l-t.g_{l+1}$ contains only terms of degree atmost $2m$ in it's support.\\

\noindent From the above claim the $S$-polynomials of any two polynomials in $\mathcal{F}$ have to be reduced by polynomials in $\mathcal{F}$ since the support of all polynomials in $\mathcal{G}$ have degree at least $2m+1$. Also, we know that $S(f,g) \xrightarrow{\mathcal{F}'} 0$. This proves that $\mathcal{F}$ is a Gr\"oner basis with respect to $\prec$.
  
\end{proof}


\subsection{Algorithm}
Consider the polynomial ring $K[x_1, \ldots ,x_n]$. Let $\mathcal{F}=\{f_1, \ldots ,f_s\}$ be the input set of polynomials to the Gr\"obner basis detection for zero dimensional ideals. If $\mathcal{F}$ is a Gr\"obner basis of a zero dimensional ideal then with respect to each $i \in \{1, \ldots ,n\}$ there exists a polynomial $f_{j_i}$ where  $j_i \in \{1, \ldots ,s\}$ such that $lt(f_{j_i})={x_i}^{\alpha_{j_i}}$. Hence, we need to find a term order $\prec$ such that the above property holds. 
\par Let $\mathcal{F}=\mathcal{F}_1 \cup \mathcal{F}_2$ such that $\mathcal{F}_1$ be the set of polynomials where each polynomial contains at least one pure power in its support  and $\mathcal{F}_2$ be the set of polynomials where each polynomial does not contain any pure power in its support. Let $f_i \in \mathcal{F}_1$ be written as:\\
$$f_i={g_i}+{h_i}$$
such that $g_i$ is a polynomial containing only pure powers in its support and $h_i$ is a polynomial containing no pure power in its support such that if $t \in \mathrm{Supp}(g_i)$ then $t$ is of the form $t={x_j}^{\alpha_{ji}}$ for some $j \in \{1, \ldots ,n\}$ and for all $j$, ${x_j}^{\alpha_{ji}} \notin \mathrm{Supp}(h_i)$. Let the set of all $g_i$'s be ${\mathcal{G}_1}=\{{g_1}, \ldots ,{g_r}\}$ such that none of $g_i$ is nonzero. We can safely assume that $r \geq n$ since if $r < n$ then this violates the condition mentioned in Definition \ref{zerodim} and hence, $\mathcal{F}$ cannot be a Gr\"obner basis of a zero dimensional ideal.\\ 

\noindent \textit{Algorithm}:\\
\textbf{Step 1}: Consider a $n$-subset of ${\mathcal{F}_1}$. Consider the corresponding subset $S$ in ${\mathcal{G}_1}$.\\
\textbf{Step 2}: Compute a unique term order $\prec$ such that the leading terms of polynomials in $S$ are mutually disjoint.\\
\textbf{Step 3}: With respect to $\prec$, test whether $\mathcal{F}$ is a Gr\"obner basis.\\
\textbf{Step 4}: If $\mathcal{F}$ is a Gr\"obner basis \textit{w.r.t} $\prec$, then return ("Yes" and term order $\prec$) else repeat Step 1.\\
\textbf{Step 5}: Return "No".\\

\noindent Correctness of the algorithm: Assume that the algorithm returns ``Yes" then for a particular $n$-subset $S$ of ${\mathcal{G}_1}$, the leading terms of all the polynomials in $S$ are mutually coprime. This can happen only if the leading terms of polynomials in $\mathcal{G}_1$ are pure powers. Also, in Step 4 we check whether $\mathcal{F}$ is a Gr\"obner basis \textit{w.r.t} $\prec$. Hence with respect to $\prec$, $\mathcal{F}$ is a Gr\"obner basis such that for each indeterminate $x_i$ there exists a polynomial such that the leading term of that polynomial is a pure power in $x_i$. In other words, there exists a term order such that $\mathcal{F}$ is a Gr\"obner basis with respect to that term order. 
\par Conversely assume $\mathcal{F}$ is a Gr\"obner basis of a zero dimensional ideal with respect to a term order $\prec$. Then there exists a subset $\{f_{i_1}, \ldots ,f_{i_n}\}$ of $\mathcal{F}$ such that $f_{i_j}=x_j^{\alpha_{ji_j}}$. For this subset $S$, a term order $\prec'$ is detected in Step 2 such that the leading terms of all the polynomials in $S$ are mutually coprime. For this subset $S$, finding a term order such that the leading terms of all the polynomials in $S$ are mutually coprime is equivalent to finding a permutation $\sigma$ which can be realized by a term order. Now, consider the following lemmas from~\cite{GBD}: first, let $F=\{f_1, \ldots ,f_n\}$ such that $f_i=X_1^{a_{i1}}+ \cdots +X_n^{a_{in}}$ for all $i \in \{1, \ldots ,n\}$.

\begin{lemma}
\label{existtermorder}
A permutation $\sigma$ cannot be realized by a term order if there is another permutation $\rho$ such that
$$\Pi_{i=1}^{n} a_{i \rho (i) } \geq  \Pi_{i=1}^{n} a_{i \sigma (i) } $$
\end{lemma}

\begin{lemma}
\label{detecttermorder}
For any monomial $X^{\beta_i} \neq X^{\alpha_i}$ occuring in $f_i$ consider the difference vector $\alpha_i-\beta_i$ and let $\Gamma$ be the matrix whose rows are all these vectors for all $i$. There exists a term order $w$ such that $LT_{w}(f_i)=X^{\alpha_i}$ for $i=1, \ldots ,n$ if and only if the linear system of inequalities $\Gamma w > 0$, $w > 0$ has a solution. Moreover if a solution exists, there also exists a solution of binary size which is polynomial in the binary size of the sparsely encoded input polynomials.   
\end{lemma}
 
From Lemma \ref{detecttermorder}, a term order $\prec$ is detected in Step 2 such that the leading terms of all the polynomials in $S$ are mutually coprime and it is unique from Lemma \ref{existtermorder}. Since $S$ is a $n$-subset, all the leading terms of $S$ are pure powers such that no two leading terms are pure powers of the same indeterminate. Consequently in Step 4, $\mathcal{F}$ is verified to be a Gr\"obner basis with respect to $\prec$ and hence returns "Yes" contradicting the assumption that the algorithm returned "No".\\   
\noindent Analysis of running time of the algorithm: Step 2 and 3 take $f(n)$ time where $f(n)$ is a polynomial in $n$. But the number of iterations of the algorithm is equal to number of all possible $n$-subsets of ${\mathcal{F}_1}$. Hence, the number of iterations can be upper bounded by $s^n$. Hence, running time of the algorithm is $O({s^n}f(n))$. Note that if the number of indeterminates was a constant then the algorithm runs in time polynomial in the number of input polynomials.   

\section{Concluding remarks}
In this paper, we analyze the complexity of the Gr\"obner basis detection problem for the case of zero-dimensional ideals and show that the problem is NP-hard. We also propose an algorithm to solve the GBD problem for the zero-dimensional case which runs in polynomial time if the number of indeterminates is a constant.

\end{document}